\documentclass[letterpaper, 10 pt, conference]{ieeeconf}  

\IEEEoverridecommandlockouts                              

\usepackage{graphics} 
\usepackage{epsfig} 
\usepackage{mathptmx} 
\usepackage{times} 
\usepackage{amsmath} 
\usepackage{amssymb}  
\usepackage{mathtools}
\usepackage{nicefrac}
\usepackage{multirow} 
\usepackage{tikz}

\usepackage{hyperref}

\newtheorem{theorem}{Theorem}
\newtheorem{definition}{Definition}
\newtheorem{lemma}{Lemma}

\newtheorem{remark}{Remark}
\newtheorem{assumption}{Assumption}

\newcommand{\x}{\hat x}
\newcommand{\e}{\hat e}
\newcommand{\NN}{\mathcal{N}}
\newcommand{\RR}{\mathbb{R}}
\newcommand{\EE}{\mathbb{E}}
\newcommand{\ZZ}{\mathbb{Z}}

\title{\LARGE \bf
LQG for Constrained Linear Systems: \\ Indirect Feedback Stochastic MPC with Kalman Filtering
}

\author{Simon Muntwiler$^{\star,1}$, Kim P. Wabersich$^{\star,1}$, Robert Miklos$^{2}$, and Melanie N. Zeilinger$^{1}$
\thanks{$^\star$The first two authors contributed equally to this work.}
\thanks{$^{1}$Simon Muntwiler, Kim P. Wabersich, and Melanie N. Zeilinger are members of the Institute for Dynamic Systems and Control, ETH Zurich, Zurich, Switzerland {\tt\small $\{$simonmu, wkim, mzeilinger$\}$@ethz.ch}.}%
\thanks{$^{2}$Robert Miklos is an independent scholar {\tt\small robert.mik @hotmail.com}.}%
\thanks{This work was supported by the Bosch Research Foundation im Stifterverband.}
}

\begin{document}

\maketitle
\thispagestyle{empty}
\pagestyle{empty}

\begin{tikzpicture}[overlay, remember picture]
\node[anchor=south,yshift=8pt] at (current page.south) {\fbox{\parbox{0.99\textwidth}{\footnotesize \textbf{Published in: 2023 European Control Conference (ECC). DOI: 10.23919/ECC57647.2023.10178356.}\\ 
			\textcopyright \space 2023 IEEE. Personal use of this material is permitted. Permission from IEEE must be obtained for all other uses, in any current or future media, including reprinting/republishing this material for advertising or promotional purposes, creating new collective works, for resale or redistribution to servers or lists, or reuse of any copyrighted component of this work in other works.}}};
\end{tikzpicture}

\begin{abstract}

We present an output feedback stochastic model predictive control (SMPC) approach for linear systems subject to Gaussian disturbances and measurement noise and probabilistic constraints on system states and inputs.
The presented approach combines a linear Kalman filter for state estimation with an indirect feedback SMPC, which is initialized with a predicted nominal state, while feedback of the current state estimate enters through the objective of the SMPC problem.
For this combination, we establish recursive feasibility of the SMPC problem due to the chosen initialization, and closed-loop chance constraint satisfaction thanks to an appropriate tightening of the constraints in the SMPC problem also considering the state estimation uncertainty.
Additionally, we show that for specific design choices in the SMPC problem, the unconstrained linear-quadratic-Gaussian (LQG) solution is recovered if it is feasible for a given initial condition and the considered constraints.
We demonstrate this fact for a numerical example, and show that the resulting output feedback controller can provide non-conservative constraint satisfaction.

\end{abstract}

\section{Introduction}\label{sec:introduction}

Model predictive control (MPC) is a well established control approach allowing for guaranteed satisfaction of constraints, which is crucial in safety-critical applications.
Model uncertainties can be addressed via robust or stochastic techniques.
While robust MPC approaches ensure satisfaction of constraints for \emph{all} possible disturbance values, stochastic MPC (SMPC) approaches commonly ensure satisfaction of constraints in probability (see, e.g.,~\cite{kouvaritakis2016model}).
Using an SMPC approach therefore often allows for reduced conservatism by considering the distribution of the uncertainty, including distributions with unbounded support such as Gaussian distributions.
For many MPC approaches, a key assumption is the access to full and exact measurements of the current system state, while often only noisy output measurements are available in practice.
Therefore, an MPC scheme is often used in connection with a state estimator, resulting in an output feedback control approach.
The uncertainty in the state estimates is, however, generally not considered when designing the MPC problem, even though it has a large influence on the prediction accuracy (see, e.g.,~\cite{Sehr2016}) and thus results in a loss of constraint satisfaction guarantees in general.
In this paper, for the case of linear systems subject to probabilistic constraints and Gaussian uncertainties, we investigate the combination of a Kalman filter (KF)~\cite{Kalman1960} for state estimation with an SMPC scheme.
The proposed approach provides guaranteed closed-loop constraint satisfaction by considering the state estimation uncertainty in the controller design.

\subsubsection*{Related Work}
A widely studied problem in the context of output feedback control is linear-quadratic-Gaussian (LQG) control for unconstrained linear systems~\cite{Lindquist2006}.
It is well known that the optimal LQG solution is the combination of a KF with a linear-quadratic regulator (LQR), which can be designed independently due to the separation principle.
Note that in the case of constrained systems, the corresponding optimal control problem can in general only be solved approximatively, e.g., using a combination of nominal MPC with a KF as in~\cite{Sehr2016,Yan2005}.
Robust satisfaction of constraints in an output feedback setting can be achieved by combining tube MPC with a linear observer~\cite{Mayne2006}.
An output feedback SMPC approach with guaranteed satisfaction of probabilistic constraints was introduced in~\cite{Cannon2012} for linear systems subject to bounded disturbances and measurements noise.
An approach for linear systems subject to unbounded uncertainties was introduced in~\cite{Farina2015}, where at each time step, the resulting SMPC problem optimizes over sequences of nominal states and inputs, and linear observer and controller gains.
This results in high online computational complexity, which is partly overcome by introducing two simplifying approximations.
However, the approach neither allows for the use of a standard KF nor to consider the most recent measurement to obtain a state estimate.
Additionally, it is challenging to establish closed-loop chance constraint satisfaction guarantees.
In~\cite{Mark2020}, a distributed output feedback SMPC approach for linear distributed systems was introduced, which relies on tightened deterministic constraints obtained using probabilistic reachable sets (PRS).
Using the same conditions as introduced in~\cite{Hewing2019a} for the state feedback case, the approach can be shown to ensure satisfaction of chance constraints in closed-loop.
However, this requires the use of symmetric PRS which can lead to additional conservatism, e.g., in the case of half-space constraints.

Assuming full state measurements are available, the recently proposed indirect feedback SMPC (IF-SMPC) approach~\cite{Hewing2018} ensures closed-loop chance constraint satisfaction, which is in general difficult, especially in the case of unbounded disturbances (cf. the discussion in~\cite{Kohler2022}).
The IF-SMPC approach relies on a split into a nominal and error system, and a corresponding reformulation of the chance constraints as equivalent tightened deterministic constraints on nominal states.
The IF-SMPC problem is initialized with a predicted nominal state which is independent of the disturbance realizations, and thereby ensures recursive feasibility even in the case of unbounded disturbances.
Feedback enters the MPC optimization problem only through the objective, which depends on predicted true system states.
In case only disturbance samples are available, arguments from scenario optimization can be used to obtain the constraint tightening~\cite{Hewing2019}, which can be easily extended to the case of distributed systems with coupled dynamics~\cite{Muntwiler*2020}.
The performance of IF-SMPC was compared to the one of a more classical direct feedback SMPC approach in~\cite{Hewing2020}, which leads in general to more conservative constraint satisfaction.
It was also shown in~\cite{Hewing2020}, that an appropriate design of IF-SMPC recovers the solution of an infinite horizon linear quadratic regulator (LQR) in case it is feasible.
Ideas from indirect and direct feedback SMPC were recently combined through the introduction of an interpolating initial state constraint~\cite{Kohler2022}.

\subsubsection*{Contribution}
In this work, we propose the combination of a KF with an IF-SMPC for output feedback control of linear systems subject to individual half-space chance constraints as an approximate solution to the corresponding constrained LQG problem (Section~\ref{sec:IF-SMPC_and_KF}).
This results in a simple design, and the benefits of an IF-SMPC approach can be directly transferred to the output feedback setting.
Specifically, we show that the resulting indirect output feedback SMPC (IOF-SMPC) controller is recursively feasible even in the case of unbounded Gaussian uncertainties (Lemma~\ref{lem:recursive_feasiblity}).
An appropriate reformulation of the chance constraints as equivalent tightened deterministic constraints (Section~\ref{sec:tightening}) allows us to prove closed-loop chance constraint satisfaction (Theorem~\ref{thm:constraint_satisfaction}) in a non-conservative manner.
Additionally, the computational complexity of solving the IOF-SMPC problem is comparable to a nominal MPC.
In Section~\ref{sec:LQG}, we prove that under some specific design choice, the presented approach recovers the unconstrained LQG solution, if it is feasible. 
Finally, in Section~\ref{sec:numerical} we evaluate the performance of the presented approach in a numerical example, which demonstrates that the resulting output feedback controller can lead to non-conservative closed-loop constraint satisfaction and recovers the unconstrained LQG solution if feasible.

\section{Problem Formulation}
We consider a linear time-invariant system subject to additive disturbances and measurement noise of the form
\begin{align} \label{eq:sys}
x(k+1) & = A x(k) + Bu(k) + w_x(k), \\
y(k)   & = C x(k) + w_y(k),
\end{align}
with state $x(k)\in\RR^{n_x}$, input $u(k)\in\RR^{n_u}$, and independent and identically distributed (i.i.d.) Gaussian disturbances $w_x(k)\sim \NN(0, \Sigma^{w_x})$ taking values in $\RR^{n_x}$ and with\footnote{We use the notation $P\succ 0$ to denote the matrix $P$ being positive definite.} $\Sigma^{w_x}\succ 0$.
The states can only be estimated through the system output $y(k)\in\RR^{n_y}$, which is corrupted by i.i.d. Gaussian measurement noise $w_y(k)\sim\NN(0, \Sigma^{w_y})$ taking values in $\mathbb{R}^{n_y}$ and with $\Sigma^{w_y}\succ 0$.
The initial state is normally distributed, i.e., $x(0)\sim\mathcal{N}(\mu^{x_0},\Sigma^{x_0})$, with mean $\mu^{x_0}\in\mathbb{R}^{n_x}$ and covariance $\Sigma^{x_0}\succ 0$.
The system is subject to $n_c^x$ half-space chance constraints on the system states and $n_c^u$ half-space chance constraints on the inputs of the form
\begin{subequations}\label{eq:constraints}
	\begin{align} 
	\mathrm{Pr}({h_j^x}^\top x(k) \le 1) \ge& p_j^x,\ \forall j \in \{1,\ldots,n_c^x\} \label{eq:constraint_x}, \\
	\mathrm{Pr}({h_j^u}^\top u(k) \le 1) \ge& p_j^u,\ \forall j \in \{1,\ldots,n_c^u\} \label{eq:constraint_u},
	\end{align}
\end{subequations}
where $h_j^x \in \mathbb{R}^n$ and $h_j^u \in \mathbb{R}^m$, $p_j^x \in (0,1)$ and $p_j^u \in (0,1)$ are the desired satisfaction probabilities, and the probabilities are understood conditioned on the initial state.

We consider the following infinite horizon stochastic optimal control (SOC) problem with quadratic objective

\begin{subequations}\label{eq:SOC}
	\begin{align}
	\min_{\Pi} ~& \lim_{\bar{N}\rightarrow \infty}\frac{1}{\bar{N}} \EE \left( \|x(\bar{N})\|_P^2 + \sum_{k=0}^{\bar{N}-1} \|x(k)\|_Q^2 + \|u(k)\|_R^2 \right) \label{eq:SOC_objective} \\
	\mathrm{s.t.}  ~& \forall~k \in \{0,1,\ldots\}: \nonumber \\
	& x(k+1) = A x(k) + Bu(k) + w_x(k),\\
	& y(k) = C x(k) + w_y(k), \\
	& u(k) = \pi_k(\{y(i)\}_{i=0}^{k},\{u(i)\}_{i=0}^{k-1},\mu^{x_0},\Sigma^{x_0}), \label{eq:SOC_policy}\\
	& \mathrm{Pr}({h_j^x}^\top x(k) \le 1) \ge p_j^x,\ \forall j \in \{1,\ldots,n_c^x\}, \label{eq:SOC_p_x} \\
	& \mathrm{Pr}({h_j^u}^\top u(k) \le 1) \ge p_j^u,\ \forall j \in \{1,\ldots,n_c^u\} \label{eq:SOC_p_u} \\
	& w_x(k) \sim \mathcal{N}(0,\Sigma^{w_x}),~\mathrm{i.i.d.}, \label{eq:SOC_w} \\
	& w_y(k) \sim \mathcal{N}(0,\Sigma^{w_y}),~\mathrm{i.i.d.}, \label{eq:SOC_s} \\
	& x(0) \sim \mathcal{N}(\mu^{x_0},\Sigma^{x_0}),
	\end{align}
\end{subequations}
with $P,~Q,~R\succ 0$, and where $\Pi =\{\pi_0,\pi_{1},\ldots\}$ is a sequence of control laws using all output measurements and inputs available at the current time step, as well as the mean and covariance of the initial state distribution.
Solving~\eqref{eq:SOC} is computationally intractable because of the constraints~\eqref{eq:SOC_p_x}-\eqref{eq:SOC_p_u}, the infinite horizon, and the optimization over general policies $\Pi$.
For the case of linear systems subject to Gaussian disturbance and measurement noise, but without constraints~\eqref{eq:constraints}, the optimal solution to~\eqref{eq:SOC} is the classical LQG controller (see, e.g.,~\cite{bertsekas2011dynamic}).
Due to the well known separation principle (see, e.g.,~\cite{kwakernaak1974linear}), the optimal solution for the unconstrained infinite horizon LQG problem is a combination of a time-invariant KF and a time-invariant LQR, provided the following assumption holds.
\begin{assumption}[Stabilizability and Detectability]\label{ass:detectability}
	The pair $(A,B)$ is stabilizable and the pairs $(A,Q^{\nicefrac{1}{2}})$ and $(A,C)$ are detectable.
\end{assumption}

In the following section, we introduce a receding horizon approach to obtain an approximate solution to the constrained LQG problem, i.e., problem~\eqref{eq:SOC} with constraints~\eqref{eq:SOC_p_x}-\eqref{eq:SOC_p_u}.
\section{Indirect Feedback SMPC and Kalman Filtering} \label{sec:IF-SMPC_and_KF}
In order to approximate the solution to~\eqref{eq:SOC}, we obtain the control input $u(k)$ for system~\eqref{eq:sys} at each time step~$k$ in a two-stage approach: First, a KF~\cite{Kalman1960} is applied to obtain an estimate $\x(k)$ of the current system state using the current measurement $y(k)$ (Section~\ref{sec:state_observer}).
Second, an IF-SMPC~\cite{Hewing2018} initialized with a predicted nominal state is used to compute a control input for the system (Section~\ref{sec:IF-SMPC}).
Thereby, feedback based on the current state estimates enters indirectly through the objective of the IF-SMPC problem.
In Section~\ref{sec:anlysis}, we show that the considered IF-SMPC is recursively feasible (Lemma~\ref{lem:recursive_feasiblity}), and the chance constraints~\eqref{eq:constraints} are satisfied in closed-loop (Theorem~\ref{thm:constraint_satisfaction}), provided the state estimation uncertainty is considered when obtaining deterministic tightened constraints on the nominal system state (Section~\ref{sec:tightening}).
Furthermore, in Section~\ref{sec:LQG} below, we show that the presented two-stage approach recovers the unconstrained LQG solution if it is feasible for~\eqref{eq:SOC} (Theorem~\ref{thm:LQG}), given an appropriate design of the ingredients of the MPC problem.
We refer to our discussion in Section~\ref{sec:LQG} for further details on the relation to LQG.
\subsection{State Estimator and Error Dynamics}\label{sec:state_observer}
The state estimate $\x (k)$ at each time step $k$ is obtained as
\begin{align} \label{eq:state_observer}
	\begin{split}
		\x(k) =& A\x(k-1) + Bu(k-1) \\
		&+ L\left(y(k) - CA\x(k-1) - CBu(k-1)\right),
	\end{split}
\end{align}
initialized with $\x(0)=\mu^{x_0}$, where $L$ is the KF gain
\begin{align}\label{eq:kalman_gain}
	L =& \hat{P}C^\top(\Sigma^s+C\hat{P}C^\top)^{-1}, 
\end{align}
with $\hat{P}\succ 0$ being the positive definite solution to
\begin{align}
	\hat{P} = A(\hat{P}-\hat{P}C^\top(C\hat{P}C^\top+\Sigma^{w_y})^{-1}C\hat{P})A^\top + \Sigma^{w_x}. \label{eq:kalman_riccati}
\end{align}
Note that~\eqref{eq:kalman_gain} and~\eqref{eq:kalman_riccati} are independent of the expected initial system state $\mu^{x_0}$ and covariance matrix $\Sigma^{x_0}$, and $L$ can therefore be computed offline for all $\mu^{x_0}$ and $\Sigma^{x_0}$.

We introduce the estimation error as 
\begin{align}\label{eq:estimation_error}
		\Delta(k) \coloneqq x(k) - \x(k)
\end{align} resulting in the estimation error dynamics
\begin{align}\label{eq:estimation_error_dynamics}
	\Delta(k+1) = (A-LCA)\Delta(k) + (I- LC) w_x(k) - Lw_y(k+1),
\end{align}
initialized with $\Delta(0)=x(0)-\x(0)\sim\mathcal{N}(0,\Sigma^{x_0})$.

Furthermore, we introduce the nominal state $z(k)$ with corresponding error $e(k)$ and estimated error $\e(k)$ as
\begin{align}
	e(k) \coloneqq& x(k)-z(k), \label{eq:error} \\
	\e(k)\coloneqq& \x(k) - z(k) = \e(k) - \Delta(k). \label{eq:estiamted_error} 
\end{align}
Introducing the nominal input $v(k)$ and tube controller $K\e(k)$, results in nominal dynamics, control policy, and error dynamics of the following form
\begin{align}
	z(k+1) =& Az(k) + Bv(k), \label{eq:nominal_dynamics}\\
	u(k) =& v(k) + K\e(k), \label{eq:input} \\
	e(k+1) =& (A+BK)e(k) - BK \Delta(k) + w_x(k), \label{eq:nominal__error_dynamics}
\end{align}
initialized with $z(0)=\mu^{x_0}$ and $e(0)=x(0)-z(0)\sim\mathcal{N}(0,\Sigma^{x_0})$.
Finally, introducing $\xi(k)=\left[e(k)^\top~\Delta(k)^\top\right]^\top$, $\tilde{w}(k)=\left[w_x(k)^\top~w_y(k+1)^\top\right]^\top$, $A_K=(A+BK)$, and $A_L=A-LCA$, we obtain the combined error dynamics
\begin{align}\label{eq:overall_error_dynamics}
\xi (k+1)
=
\underbrace{\begin{bmatrix}
	A_K & -BK \\ 0 & A_L
	\end{bmatrix}}_{\tilde{A}(K)}\xi(k)
+
\underbrace{\begin{bmatrix}
	I & 0 \\ I -LC & -L
	\end{bmatrix}}_{\tilde{B}}\tilde{w}(k).
\end{align}
Note that the combined error dynamics is independent of the true system state and its distribution can thus be computed offline.

From the definition of the nominal error $e(k)$ in~\eqref{eq:error} and estimation error $\Delta(k)$ in~\eqref{eq:estimation_error}, and the initialization of both the state estimator~\eqref{eq:state_observer} and nominal dynamics~\eqref{eq:nominal_dynamics} with $\mu^{x_0}$, we have that $e(0)=\Delta(0)\sim\mathcal{N}(0,\Sigma^{x_0})$.
Together with the distributions of the disturbance $w_x(k)$ and measurement noise $w_y(k)$, we have that
\begin{align*}
\xi (0) \sim \mathcal{N} \left(0,\Sigma^\xi (0) \right), \tilde{w}(k) \sim \mathcal{N} \left(0,\Sigma^{\tilde{w}}\right),
\end{align*}
with
\begin{align*}
\Sigma^\xi (0) = \begin{bmatrix}
1 & 1 \\
1 & 1
\end{bmatrix} \otimes \Sigma^{x_0} \text{ and } \Sigma^{\tilde{w}} = \begin{bmatrix}
\Sigma^{w_x} & 0 \\
0 & \Sigma^{w_y}
\end{bmatrix},
\end{align*}
where $\otimes$ denotes the Kronecker product.
Consequently, we have that $\xi (k) \sim \mathcal{N}(0,\Sigma^\xi (k))$ with
\begin{align}\label{eq:overall_error_covariance}
\Sigma^\xi (k+1) = \tilde{A}(K) \Sigma^\xi(k) \tilde{A}(K)^\top + \tilde{B}\Sigma^{\tilde{w}}\tilde{B}^\top,
\end{align}
and $\tilde{A}(K)$ and $\tilde{B}$ as defined in~\eqref{eq:overall_error_dynamics}.
\subsection{Constraint Tightening} \label{sec:tightening}
In the following, we show how the combined error dynamics~\eqref{eq:overall_error_dynamics} and the corresponding covariance dynamics~\eqref{eq:overall_error_covariance} can be used to design deterministic constraints on the nominal state and input, which allow for satisfaction of the underlying chance constraints~\eqref{eq:constraints} for the closed-loop system when applying the receding horizon controller.

Using the definitions of $e(k)$ and $u(k)$ in~\eqref{eq:error} and~\eqref{eq:input}, respectively, we can rewrite the chance constraints as
\begin{subequations}\label{eq:constraints_z+e}
	\begin{align}
	\mathrm{Pr}({h_j^x}^\top z(k) + {h_j^x}^\top e(k) \le 1) \ge& p_j^x,\ \forall j \in \{1,\ldots,n_c^x\}, \\
	\mathrm{Pr}({h_j^u}^\top v(k) + {h_j^u}^\top K\e(k) \le 1) \ge& p_j^u,\ \forall j \in \{1,\ldots,n_c^u\}.
	\end{align}
\end{subequations}
In order to transform above chance constraints into deterministic constraints on the nominal state and input, we use the following concept of probabilistic reachable sets for the error state $e(k)$ and tube controller $K\e(k)$.
\begin{definition}[\protect {$k$-step PRS~\cite[Def. 1]{Hewing2018}}]\label{def:prs}
	Consider system
	\begin{align}\label{eq:sys_zeta}
		\zeta(k+1) = A_{\zeta}\zeta(k) + B_{\zeta}w_{\zeta}(k),
	\end{align}
	with state $\zeta(k)\in\mathbb{R}^{n_\zeta}$ initialized with $\zeta(0)\sim\mathcal{Q}_{\zeta}$, and disturbance $w_\zeta(k)\in\mathbb{R}^{m_\zeta}$ with $w_\zeta(k)\sim\mathcal{Q}_{w_\zeta}$, where $\mathcal{Q}_{\zeta}$ and $\mathcal{Q}_{w_\zeta}$ are probability distributions.
	A set $\mathcal{R}_k$ with $k \in \{0,1,\ldots\}$ is a $k$-step probabilistic reachable set ($k$-step PRS) of probability level $p$ for system~\eqref{eq:sys_zeta}  if
	\begin{align*}
	\mathrm{Pr}(\zeta(k)\in\mathcal{R}_k|\zeta(0))\ge p.
	\end{align*}\vspace{-5mm}
\end{definition}

From the overall error dynamics~\eqref{eq:overall_error_dynamics} and covariance dynamics~\eqref{eq:overall_error_covariance} we can obtain the distribution of the error state and tube controller as $e(k)\sim\mathcal{N}(0,\Sigma^{\mathrm{e}}(k))$ and $K\e(k) = K(e(k)-\Delta(k))\sim\mathcal{N}(0,\Sigma^{\mathrm{tb}}(k))$, respectively, with
\begin{subequations}\label{eq:error_covariance}
	\begin{align}
	\Sigma^{\mathrm{e}}(k) =& \begin{bmatrix}
	I & 0
	\end{bmatrix} \Sigma^\xi(k) \begin{bmatrix}
	I & 0
	\end{bmatrix}^\top, \\
	\Sigma^{\mathrm{tb}}(k) =& K\begin{bmatrix}
	I & -I
	\end{bmatrix} \Sigma^\xi(k) \begin{bmatrix}
	I & -I
	\end{bmatrix}^\top K^\top,
	\end{align}
\end{subequations}
as similarly done in~\cite[Sec. 2.2]{Farina2015}.

Using $k$-step half-space PRS for both the error state $e(k)$ and tube controller $K\e(k)$, i.e., a $k$-step PRS according to Definition~\ref{def:prs} of the form $\mathcal{R}_k = \{x|h_k^\top x \le 1\}$, we can transform~\eqref{eq:constraints_z+e} into tightened constraints for all $k \in \{0,1,\ldots\}$~\cite[Sec. 3.2]{Hewing2018}
\begin{subequations}\label{eq:constraints_nominal}
	\begin{align}
	{h_j^x}^\top z(k) \le& 1 - c_{j,k}^x,\ \forall j \in \{1,\ldots,n_c^x\}, \\
	{h_j^u}^\top v(k) \le& 1 - c_{j,k}^u,\ \forall j \in \{1,\ldots,n_c^u\},
	\end{align}
\end{subequations}
where the tightening values $c_{j,k}^x$ and $c_{j,k}^u$ are obtained as
\begin{subequations}\label{eq:tightenings}
	\begin{align}
		c_{j,k}^x =& \sqrt{\tilde{p}_j^x{h_j^x}^\top \Sigma^{\mathrm{e}}(k) h_j^x},\ \forall j \in \{1,\ldots,n_c^x\}, \\
		c_{j,k}^u =& \sqrt{\tilde{p}_j^u{h_j^u}^\top \Sigma^{\mathrm{tb}}(k) h_j^u},\ \forall j \in \{1,\ldots,n_c^u\},
	\end{align}
\end{subequations}
with $\tilde{p}_j^x=\chi_1^2(2p_j^x-1)$ and $\tilde{p}_j^u=\chi_1^2(2p_j^u-1)$, where $\chi_1^2$ is the quantile function of the chi-squared distribution with~$1$ degree of freedom.
Note that the tightening of the half-space constraints~\eqref{eq:constraints} using half-space PRS in~\eqref{eq:constraints_nominal} is non-conservative in a sense that whenever the nominal constraint~\eqref{eq:constraints_nominal} is always active, the original probabilistic constraint is satisfied at the desired probability level (compare~\cite{Lorenzen2017},~\cite[Rem. 3]{Hewing2018}).

\subsection{MPC Formulation}\label{sec:IF-SMPC}
After introducing the nominal system~\eqref{eq:nominal_dynamics} and corresponding deterministic constraints~\eqref{eq:constraints_nominal}, we now proceed to the formulation of the IOF-SMPC scheme.
As in IF-SMPC~\cite{Hewing2018}, the SMPC problem is initialized with a predicted nominal state.
Constraints are enforced on the nominal state only, and thus recursive feasibility can be easily established.
Feedback of the current system state, here the estimate $\x(k)$ of the current system state, is introduced through the objective of the SMPC problem, and thus termed indirect.
For a quadratic objective~\eqref{eq:SOC_objective} and under the considered controller class~\eqref{eq:input}, minimization of the expectation in~\eqref{eq:SOC_objective} is equivalent to minimizing the cost at the mean~\cite{Hewing2020}.
At each time step $k$, the following optimization problem is solved, initialized with the current state estimate $\hat{x}(k)$ obtained from~\eqref{eq:state_observer} and a predicted nominal state $z_1(k-1)$:
\begin{subequations}\label{eq:MPC}
	\begin{align}
	\min_{\{v_{i|k}\}_{i=0}^{N-1}, \{z_{i|k}\}_{i=0}^{N}} ~& \| \bar{x}_{N|k} \|_P^2 + \sum_{i=0}^{N-1} \|\bar{x}_{i|k}\|_Q^2 + \|v_{i|k}+K\bar{e}_{i|k}\|_R^2 \label{eq:MPC_objective} \\
	\mathrm{s.t.}  ~& \forall~i \in \{0,\ldots,N-1\}: \nonumber \\
	& \bar{x}_{i+1|k} = z_{i+1|k} + \bar{e}_{i+1|k}, \label{eq:MPC_mean_state} \\ 
	& z_{i+1|k} = Az_{i|k} + Bv_{i|k}, \\
	& \bar{e}_{i+1|k} = (A+BK)\bar{e}_{i|k}, \label{eq:MPC_error_dynamics} \\
	& {h_j^x}^\top z_{i|k} \le 1 - c_{j,k+i}^x,\ \forall j \in \{1,\ldots,n_c^x\}, \label{eq:MPC_constraint_z} \\
	& {h_j^u}^\top v_{i|k} \le 1 - c_{j,k+i}^u,\ \forall j \in \{1,\ldots,n_c^u\}, \label{eq:MPC_constraint_v}\\
	& z_{N|k} \in \ZZ_f, \label{eq:MPC_terminal_constr}\\
	& \bar{x}_{0|k} = \x(k), ~z_{0|k}=z_1(k-1), \label{eq:MPC_initial_state}\\
	& \bar{e}_{0|k}=\bar{x}_{0|k} - z_{0|k}, \label{eq:MPC_initial_error}
	\end{align}
\end{subequations}
where $\bar{x}_{i|k}$ denotes the mean of the predicted state of system~\eqref{eq:sys} at time step $i$ computed at time step $k$ and initialized with the current state estimate, i.e., $\bar{x}_{0|k}=\x(k)$, and $\bar{e}_{i|k}$ denotes the corresponding mean error.
Note that different from~\eqref{eq:nominal__error_dynamics}, the mean error prediction in~\eqref{eq:MPC_error_dynamics} does not contain predicted estimation errors, i.e., terms involving $\Delta_{i|k}$, because the estimation error dynamics~\eqref{eq:estimation_error_dynamics} is subject to zero mean disturbances and initialized with $\Delta_{0|k}=0$ at each time step~$k$.
The resulting control input to system~\eqref{eq:sys} is then defined using the optimizer of~\eqref{eq:MPC} as
\begin{align}\label{eq:MPC_input}
	u(k) = \pi_{\mathrm{MPC}}(\x(k),z(k)) = v_{0|k}^* + K(\x (k) - z(k)),
\end{align}
where $z(k)=z_1(k-1)$.
The predicted nominal state is obtained as
\begin{align}\label{eq:predicted_nominal_state}
	z_1(k) = z_{1|k}^*,
\end{align}
initialized with $z_1(-1)=\x(0)$.
The terminal set $\ZZ_f$ in~\eqref{eq:MPC_terminal_constr} is designed as follows.
\begin{assumption}[\protect {Terminal Invariance~\cite[Ass. 1]{Hewing2018}}]\label{ass:terminal_invariance}
	The terminal set $\ZZ_f$ is positively invariant for system~\eqref{eq:nominal_dynamics} under a terminal control law $\pi_f(z)$, i.e., for all $z\in\ZZ_f$ we have $Az+B\pi_f(z)\in\ZZ_f$.
	Additionally, for all $z\in\ZZ_f$ it holds that
	\begin{align*}
		{h_j^x}^\top z \le& 1 - \sqrt{\tilde{p}_j^x{h_j^x}^\top \Sigma^{\mathrm{e}}_{\infty} h_j^x},\ \forall j \in \{1,\ldots,n_c^x\}, \\
		{h_j^u}^\top \pi_f(z) \le& 1 - \sqrt{\tilde{p}_j^u{h_j^u}^\top \Sigma^{\mathrm{tb}}_{\infty} h_j^u},\ \forall j \in \{1,\ldots,n_c^u\},
	\end{align*}
	where the tightening is obtained using a probabilistic reachable set (PRS)~\cite[Def. 2]{Hewing2018}, i.e., a set $\mathcal{R}_k$ according to Definition~\ref{def:prs} with $k\rightarrow\infty$, with covariance matrices
	\begin{subequations}\label{eq:error_covariance_infty}
		\begin{align}
		\Sigma^{\mathrm{e}}_\infty =& \begin{bmatrix}
		I & 0
		\end{bmatrix} \Sigma^\xi_\infty \begin{bmatrix}
		I & 0
		\end{bmatrix}^\top, \\
		\Sigma^{\mathrm{tb}}_\infty =& K\begin{bmatrix}
		I & -I
		\end{bmatrix} \Sigma^\xi_\infty \begin{bmatrix}
		I & -I
		\end{bmatrix}^\top K^\top,
		\end{align}
	\end{subequations}
	and $\Sigma_\infty^\xi$ is obtained as the solution to the Lyapunov equation 
	\begin{align*}
		\Sigma_\infty^\xi = \tilde{A}(K)\Sigma_\infty^\xi \tilde{A}(K)^\top + \tilde{B}\Sigma^{\tilde{w}}\tilde{B}^\top.
	\end{align*}\vspace{-5mm}
\end{assumption}
\begin{remark}[Finite-horizon Problems]
	In practice, in case a finite-horizon control task (i.e., the SOC problem~\eqref{eq:SOC} with potentially long but finite horizon $\bar{N}$) is considered, e.g., when the objective is to control a system from one steady-state to another one within finite time steps, a time-varying KF with gain $L(k)$ can be considered instead of~\eqref{eq:state_observer}.
	This allows for an improved estimation performance, while recursive feasibility and constraint satisfaction as discussed in Section~\ref{sec:anlysis} below can still be established by appropriately adapting the tightening in~\eqref{eq:MPC_constraint_z}-\eqref{eq:MPC_constraint_v}. 
\end{remark}
\subsection{Recursive Feasibility and Constraint Satisfaction}\label{sec:anlysis}
In this section, we establish recursive feasibility and closed-loop chance constraint satisfaction for the presented indirect output feedback SMPC approach~\eqref{eq:MPC}.
As common in IF-SMPC, recursive feasibility follows directly from standard MPC results given a suitable terminal set (Assumption~\ref{ass:terminal_invariance}) due to the initialization of the nominal state using the previously predicted nominal state~\eqref{eq:predicted_nominal_state} and the constraints only acting on the nominal state and input (compare \cite[Thm.~1]{Hewing2018}).
\begin{lemma}[Recursive Feasibility]\label{lem:recursive_feasiblity}
	Consider system~\eqref{eq:sys} under the control law~\eqref{eq:MPC_input} with state estimates obtained from~\eqref{eq:state_observer}, and let Assumptions~\ref{ass:terminal_invariance} hold.
	If the optimization problem~\eqref{eq:MPC} is feasible at the initial time step for $\x(0) = \mu^x$ and $z_1(-1) = \mu^x$, then it is recursively feasible, i.e., it is feasible for all time steps $k\in\{0,1,\ldots\}$.
\end{lemma}

Recursive feasibility of~\eqref{eq:MPC} directly implies satisfaction of the constraints~\eqref{eq:constraints} in closed-loop, due to the constraint tightening in~\eqref{eq:constraints_nominal} based on $k$-step PRS according to Definition~\ref{def:prs}.
\begin{theorem}[Closed-loop Constraint Satisfaction]\label{thm:constraint_satisfaction}
	Consider system~\eqref{eq:sys} under the control law~\eqref{eq:MPC_input} with state estimate obtained from~\eqref{eq:state_observer}, and let Assumption~\ref{ass:terminal_invariance} hold. The resulting closed-loop states $x(k)$ and inputs $u(k)$ satisfy the chance constraints~\eqref{eq:constraints} for all $k\in\{0,1,\ldots\}$.
\end{theorem}
\begin{proof}
	The proof follows similar arguments as the proof of~\cite[Thm. 2]{Hewing2018}.
	Due to the linear evolution of the combined error system~\eqref{eq:overall_error_dynamics} and by design of the constraint tightening in~\eqref{eq:tightenings} we have that $\mathrm{Pr}({h_j^x}^\top e(k) \le c_{j,k}^x) \ge p_j^x$ for all $k\in\{0,1,\ldots\}$ and $j\in\{1,\ldots,n_c^x\}$.
	From recursive feasibility (Lemma~\ref{lem:recursive_feasiblity}) and the definition of the nominal constraint~\eqref{eq:MPC_constraint_z} enforced in~\eqref{eq:MPC} it follows that $z(k)=z_{0|k}^*=z_1(k-1)$ with ${h_j^x}^\top z_1(k-1) \le 1 - c_{j,k}^x$ for all $j\in\{1,\ldots,n_c^x\}$.
	Using the definition of the nominal error, we have that $x(k) = z(k) + e(k)$ and it directly follows that the constraint~\eqref{eq:constraint_x} is satisfied for all $k\in\{0,1,\ldots\}$.
	The same arguments hold for the input constraints~\eqref{eq:constraint_u}.
\end{proof}

Note that in contrast to~\cite[Thm. 2]{Hewing2018}, Theorem~\ref{thm:constraint_satisfaction} ensures closed-loop constraint satisfaction in the absence of exact state measurements, i.e., under additional uncertainties in the corresponding state estimates.

\section{Relation to Unconstrained LQG}\label{sec:LQG}

In this section, we establish the relation of the presented combination of a KF~\eqref{eq:state_observer} and an IOF-SMPC~\eqref{eq:MPC}, with classical unconstrained LQG, under a special choice of the terminal cost weighting $P$ and the tube controller $K$, and without terminal constraint~\eqref{eq:MPC_terminal_constr}. 

\begin{assumption}[Terminal Cost Weighting]\label{ass:terminal_cost}
	The weighting $P$ of the terminal cost in~\eqref{eq:SOC} and~\eqref{eq:MPC} satisfies the Riccati equation
	\begin{align}\label{eq:LQR_riccati}
		P = A^\top PA - A^\top PB(R + B^\top P B)^{-1}B^\top P A + Q.
	\end{align}
	Additionally, the tube controller $K$ in~\eqref{eq:MPC} and~\eqref{eq:MPC_input} is chosen as the LQR controller $K_{\mathrm{LQR}}$ with
	\begin{align}\label{eq:LQR_gain}
	K_{\mathrm{LQR}} = (R+B^\top P B)^{-1}B^\top P A.
	\end{align}%
\end{assumption}

In the unconstrained case, i.e., without constraints~\eqref{eq:SOC_p_x}-\eqref{eq:SOC_p_u}, the optimal solution to~\eqref{eq:SOC} is given by $\pi_{\mathrm{LQR}}(\x(k)) = K_{\mathrm{LQR}} \x (k)$ with $K_{\mathrm{LQR}}$ from~\eqref{eq:LQR_gain} and $\x (k)$ as the KF estimate resulting from~\eqref{eq:state_observer}.

In the following theorem, we establish that the proposed combination of a KF~\eqref{eq:state_observer} with an IOF-SMPC~\eqref{eq:MPC} without terminal constraint~\eqref{eq:MPC_terminal_constr} recovers the optimal solution of the unconstrained LQG problem, i.e., the SOC problem~\eqref{eq:SOC} without constraints~\eqref{eq:SOC_p_x}-\eqref{eq:SOC_p_u}, if it is feasible.

\begin{theorem}\label{thm:LQG}
	Consider the IOF-SMPC problem~\eqref{eq:MPC} without terminal constraint~\eqref{eq:MPC_terminal_constr} and
	let Assumptions~\ref{ass:detectability} and~\ref{ass:terminal_cost} hold.
	Furthermore, let $\pi_k(\{y(i)\}_{i=0}^{k},\{u(i)\}_{i=0}^{k-1},\mu^{x_0},\Sigma^{x_0}) = K_{\mathrm{LQR}}\x (k)$, with $\x (k)$ resulting from~\eqref{eq:state_observer} and $L$ equal to the KF gain as obtained in~\eqref{eq:kalman_gain}, be the unique feasible and optimal solution to the SOC problem~\eqref{eq:SOC}.
	Then, $\pi_{MPC}(\x(k),z(k)) = K_{\mathrm{LQR}}\x(k)$ for all $k\ge 0$, i.e., the optimal solution to~\eqref{eq:SOC} is recovered by~\eqref{eq:MPC}-\eqref{eq:MPC_input}.
\end{theorem}
\begin{proof}
	The proof follows similar arguments as the proof of~\cite[Lem. 1]{Hewing2020}.
	The claim is proven by showing that
	\begin{align}\label{eq:proof_input}
	u_{i|k}^* = v_{i|k}^* + K_{\mathrm{LQR}}\bar{e}_{i|k}^* = K_{\mathrm{LQR}}\bar{x}_{i|k}^*
	\end{align}
	is the unique optimal and feasible solution in~\eqref{eq:MPC} at each time step $k\in\{0,1,\ldots\}$ and for each prediction step $i\in\{0,\ldots,N-1\}$.
	The claim then follows directly from the initialization of~\eqref{eq:MPC}, i.e., $u(k) = K_{\mathrm{LQR}}\x(k)$ due to constraints~\eqref{eq:MPC_initial_state} and~\eqref{eq:MPC_initial_error}.
	Due to the quadratic cost~\eqref{eq:MPC_objective} and the choice of $P$ according to Assumption~\ref{ass:terminal_cost}, we have that~\eqref{eq:proof_input} is the unique optimal solution when no constraints are active, i.e., given it is feasible in~\eqref{eq:MPC}.
	It remains to show that~\eqref{eq:proof_input} is always feasible in~\eqref{eq:MPC}.
	From~\eqref{eq:proof_input} and~\eqref{eq:MPC_mean_state} we have that $v_{i|k}^* = K_{\mathrm{LQR}}z_{i|k}^*$, which implies that at each time step $k\in\{0,1,\ldots\}$  the predicted nominal state evolves according to $z_{i+1|k}^* = (A+BK_{\mathrm{LQR}})z_{i|k}^*$ initialized with $z_{i|k}^*=z(k)$.
	Consequently, the closed-loop nominal states evolves according to $z(k+1) = (A+BK_{\mathrm{LQR}})z(k)$.
	Thus, for each time step $k\in\{0,1,\ldots\}$ we have that $z_{i|k}^*=z(k+i)$ for all prediction steps $i\in\{0,\ldots,N-1\}$.
	Similarly, we have that $v_{i|k}^*=v(k+i)$.
	Due to the non-conservative constraint tightening using half-space PRS, we have that the closed-loop nominal state $z(k)$ and input $v(k)$ satisfy for all $k\in\{0,1,\ldots\}$
	\begin{align*}
		&{h_j^x}^\top z(k) \le 1 - c_{j,k}^x \Leftrightarrow \mathrm{Pr}({h_j^x}^\top x(k) \le 1) \ge p_j^x, \forall j \in \{1,\ldots,n_c^x\},  \\
		&{h_j^u}^\top v(k) \le 1 - c_{j,k}^u \Leftrightarrow \mathrm{Pr}({h_j^u}^\top u(k) \le 1) \ge p_j^u, \forall j \in \{1,\ldots,n_c^u\},
	\end{align*}
	where the right hand side is satisfied due to the assumption of the LQG controller being feasible in~\eqref{eq:SOC}.
	Since the open-loop and closed-loop nominal states and inputs evolve equivalently, it follows that the constraints~\eqref{eq:MPC_constraint_z} and~\eqref{eq:MPC_constraint_v} are satisfied and thus~\eqref{eq:proof_input} is feasible in~\eqref{eq:MPC}, which concludes the proof.
\end{proof}

Note that problem~\eqref{eq:MPC} without terminal constraint~\eqref{eq:MPC_terminal_constr} is in general \emph{not} recursively feasible, and consequently also does not ensure closed-loop chance constraint satisfaction.
However, assuming the unconstrained LQG solution is feasible with respect to~\eqref{eq:SOC}, Theorem~\ref{thm:LQG} establishes that~\eqref{eq:proof_input} is the unique feasible and optimal solution of the IOF-SMPC problem~\eqref{eq:MPC} at each time step~$k$ and prediction step~$i$, which shows that recursive feasibility is satisfied in this specific case.

\section{Numerical Example} \label{sec:numerical}

\begin{figure}
	\centering
	\includegraphics[width=\columnwidth]{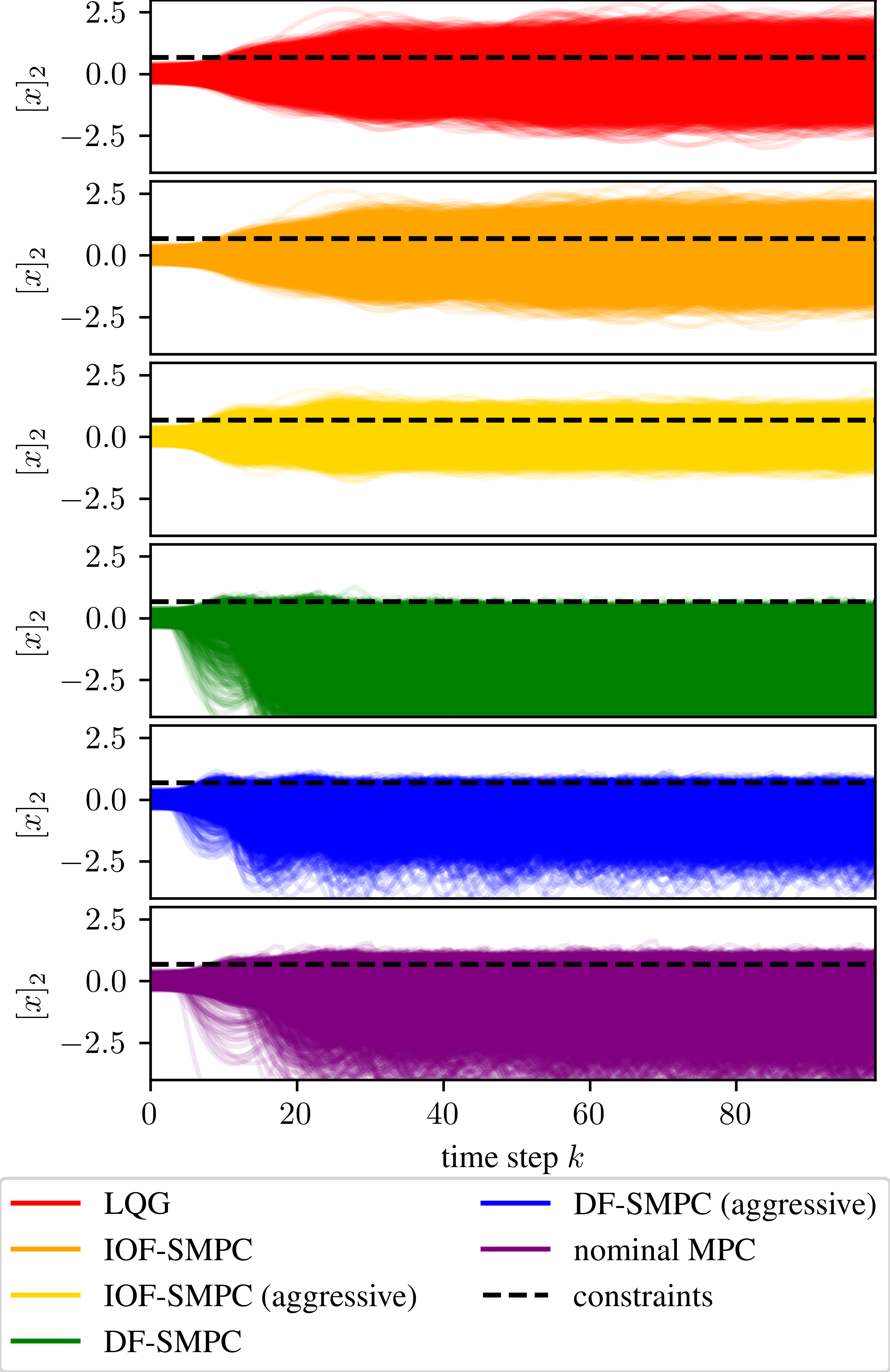}
	\caption{Simulation results for all six different output feedback controllers with zero mean initial condition.}
	\label{fig:zero_initial}
\end{figure}
In this section, we present a numerical example which shows the benefits of an indirect-feedback compared to a direct-feedback SMPC approach, as similarly done in~\cite{Hewing2020} for the case of full state measurements.
We also show that the LQG solution is recovered in case it is feasible (compare Theorem~\ref{thm:LQG}).
For this purpose, we revisit the four dimensional integrator chain from~\cite[Sec. IV]{Hewing2020} with additional output equation.
We consider a system of the form~\eqref{eq:sys} with 
\begin{align*}
	A =& \begin{bmatrix}
	1 & T_s & \nicefrac{T_s^2}{2!} & \nicefrac{T_s^3}{3!} \\
	0 & 1 & T_s & \nicefrac{T_s^2}{2!} \\
	0 & 0 & 1 & T_s \\
	0 & 0 & 0 & 1 
	\end{bmatrix},\ &&B= \begin{bmatrix}
	\nicefrac{T_s^4}{4!} \\
	\nicefrac{T_s^3}{3!} \\
	\nicefrac{T_s^2}{2!} \\
	T_s
	\end{bmatrix}, \\
	C =& \begin{bmatrix}
	1 & 0 & 0 & 0
	\end{bmatrix},\ &&
\end{align*}
and $T_s = 0.1$.
The system is subject to i.i.d. disturbances $w_x(k)$ and measurement noise $w_y(k)$ with $\Sigma^{w_x} = BB^\top$ and $\Sigma^{w_y} = 0.1^2$, and an initial state distribution with mean $\mu^{x_0}$ and covariance matrix $\Sigma^{x_0} = 0.1^2 I$.
The weighting in the quadratic cost in~\eqref{eq:SOC_objective} and~\eqref{eq:MPC_objective} is chosen as $Q=I$ and $R=0.1$ with $P$ chosen according to~\eqref{eq:LQR_riccati}.
Finally, we select a single half-space constraint on the systems state~\eqref{eq:constraint_x} with
\begin{align*}
	{h^x}^\top = \begin{bmatrix}
	0 & \nicefrac{1}{\sqrt{[\Sigma_\infty]_{2,2}}} & 0 & 0
	\end{bmatrix},\ p = 84\%,
\end{align*}
with $[\Sigma_\infty]_{2,2}$ being the element in the second row and column of the solution to the Lyapunov equation
\begin{align*}
	\Sigma_\infty = \tilde{A}(K_{\mathrm{LQR}})\Sigma_\infty \tilde{A}(K_{\mathrm{LQR}})^\top + \tilde{B}\Sigma^{\tilde{w}}\tilde{B}^\top,
\end{align*}
with $K_{\mathrm{LQR}}$ obtained from~\eqref{eq:LQR_gain}.
Note that this specific choice of the chance constraint ensures that it is asymptotically satisfied by an LQG controller (compare~\cite[Sec. IV]{Hewing2020}).
The numerical example was implemented in python using CasADi~\cite{Andersson2019} and the solver IPOPT~\cite{Waechter2005}.
The implementation is available online\footnote{\href{https://doi.org/10.3929/ethz-b-000608224}{https://doi.org/10.3929/ethz-b-000608224}}.

\begin{table*}
	\centering
	\caption{Average cost and empirical constraint violation}
	\label{tab:results}
	\begin{tabular}{c|cc|cc}
		\hline
		\multirow{2}{*}{Controller} &  \multicolumn{2}{c}{Zero Mean Initial Condition} & \multicolumn{2}{|c}{Non-Zero Mean Initial Condition} \\
		\cline{2-5}
		& Average Cost & Empirical Violation  & Average Cost & Empirical Violation \\
		\hline
		LQG & 377.7 & 16.1\% & 439.2 & 35.3\% \\
		IF-SMPC & 377.7 & 16.1\% & 481.7 & 16.1\% \\
		IF-SMPC (aggressive) & 2316.5 & 6.4\% & 2522.5 & 10.6\% \\
		DF-SMPC & 554981.6 & 0.3\% & 695383.0 & 2.5\% \\
		DF-SMPC (aggressive) & 509743.4 & 0.7\% & 780533.9 & 1.8\% \\
		Nominal MPC & 146663.9 & 5.5\% & 344724.2 & 9.6\%
	\end{tabular}
\end{table*}

We compare six different control pipelines: A classical unconstrained LQG, our IOF-SMPC approach with LQR controller as tube controller, our IOF-SMPC with an aggressive tube controller (an LQR with $Q = \mathrm{diag}(1,100,1,1)$ and $R=0.1$), a direct feedback SMPC (DF-SMPC) as discussed in~\cite{Hewing2020} with both the LQR and an aggressive version as tube controller, and a nominal MPC approach.
For the DF-SMPC approaches and the nominal MPC, the control problem is initialized with the current state estimate resulting from~\eqref{eq:state_observer} if feasible, or with the predicted (nominal) state otherwise.
In all cases, the current state estimates are obtained using a KF~\eqref{eq:state_observer}.
In Fig.~\ref{fig:zero_initial}, we show $10^5$ simulations for each of the control approaches with $\mu^{x_0} = 0$, i.e., an initial state distributed according to a zero mean Gaussian distribution.
The corresponding average closed-loop costs and empirical constraint violations for all six controllers are shown in column two and three of Table~\ref{tab:results}.
It is visible that the LQG and our IOF-SMPC with LQR tube controller lead to the same closed-loop results, average cost and level of constraint satisfaction.
Since the assumptions of Theorem~\ref{thm:LQG} hold in this case, the performance of the LQG is recovered by our IOF-SMPC approach with LQR tube controller.
This is not the case for the IOF-SMPC with aggressive tube controller, which thus leads to a different closed-loop behavior compared to the LQG.
The IOF-SMPC with aggressive tube controller still satisfies the chance constraints, but at a higher probability level than specified.
In the case of both DF-SMPC approaches and the nominal MPC, conservative satisfaction of the chance constraints lead to very large closed-loop cost compared to an LQG and both IOF-SMPC approaches.
Note that the empirical constraint violation is obtained by calculating the violation probability over all $10^5$ trajectories at each time step, and then maximizing over the time steps to find the maximum empirical violation.
This introduces a slight bias towards larger numbers for a finite number of samples or due to small numerical errors, and can thus explain that the empirical violation is slightly larger than specified for both the LQG and the IOF-SMPC with LQR tube controller.

We additionally show the resulting average cost and empirical violation for $10^5$ simulations with non-zero mean initial state distribution, i.e., with $\mu^{x_0} = \begin{bmatrix}
-1.5 & 0 & 0 & 0
\end{bmatrix}^\top$,
in column four and five of Table~\ref{tab:results}.
It can be seen that the LQG controller is not feasible in this case, resulting in a large empirical constraint violation.
Both IOF-SMPC again (empirically) satisfy the chance constraints, while the DF-SMPC and nominal MPC approaches result in conservative constraint satisfaction and large average cost.

\section{Conclusion}

In this paper, we presented an output feedback SMPC approach for linear systems with Gaussian uncertainties subject to probabilistic chance constraints.
The presented approach combines a linear KF with an indirect feedback SMPC, where the state estimation uncertainty is considered during the design phase.
The resulting indirect output feedback SMPC (IOF-SMPC) controller is shown to be recursively feasible and ensures that the chance constraints are satisfied in closed-loop.
Additionally, under appropriate design choices, the solution of the IOF-SMPC is equivalent to the solution to the unconstrained LQG problem, in case it is feasible.





\section*{DATA AVAILABILITY STATEMENT}
The code to reproduce the numerical example in Section~\ref{sec:numerical} is available in the ETH Zurich Research Collection, \href{https://doi.org/10.3929/ethz-b-000608224}{https://doi.org/10.3929/ethz-b-000608224}.


\section*{ACKNOWLEDGMENT}

The authors would like to thank Johannes Köhler for the discussions on the topic and the feedback on an initial draft, and Ueli Wechsler for sparking our interest in this topic.


\bibliographystyle{IEEEtran}
\bibliography{bibliography}

\end{document}